\documentclass[12pt]{iopart}
\usepackage{iopams}  
\expandafter\let\csname equation*\endcsname\relax
\expandafter\let\csname endequation*\endcsname\relax
\usepackage{amsmath}
\usepackage{amsfonts}
\usepackage{amscd}
\usepackage{amsthm}
\usepackage{setspace}
\usepackage{amssymb}
\usepackage{braket}
\usepackage{hyperref}
\usepackage{color}
\usepackage{graphicx}

\newtheorem{proposition}{Proposition}
\begin{document}

\title[]{Convergence condition of simulated quantum annealing with a non-stoquastic catalyst}

\author{Yusuke Kimura$^1$\footnote{Present address: Analytical quantum complexity Riken Hakubi Research Team, Riken Center for Quantum Computing (RQC), Wako, Saitama 351-0198, Japan} and Hidetoshi Nishimori$^1$ $^2$ $^3$}

\address{$^1$ International Research Frontiers Initiative, Tokyo Institute of Technology, Shibaura, Minato-ku, Tokyo 108-0023, Japan}

\address{$^2$ Graduate School of Information Sciences, Tohoku University, Sendai 980-8579, Japan}

\address{$^3$ RIKEN, Interdisciplinary Theoretical and Mathematical Sciences (iTHEMS), Wako, Saitama 351-0198, Japan}

\begin{abstract}
The Ising model with a transverse field and an antiferromagnetic transverse interaction is represented as a matrix in the computational basis with non-zero off-diagonal elements with both positive and negative signs and thus may be regarded to be non-stoquastic. We show that the local Boltzmann factors of such a system under an appropriate Suzuki-Trotter representation can be chosen non-negative and thus may potentially be simulated classically without a sign problem if the parameter values are limited to a subspace of the whole parameter space. We then derive conditions for parameters to satisfy asymptotically in order that simulated quantum annealing of this system converges to thermal equilibrium in the long-time limit.
\end{abstract}

\section{Introduction}
Quantum annealing (QA) \cite{Kadowaki1998,Brooke1999,Farhi2001,Santoro2002} has been a subject of active studies in recent years \cite{Albash2018,Hauke2020,Tanaka2017,Grant2020}. One of the outstanding theoretical problems is whether or not there exist optimization problems for which QA outperforms classical algorithms.
Toward this goal, in addition to dedicated classical algorithms designed for specific types of problems, generic approximate algorithms, metaheuristics, are often employed to benchmark the performance of QA.  They include simulated annealing \cite{Kirkpatrick1983}, simulated quantum annealing (SQA) (the path-integral Monte Carlo \cite{Suzuki}), and the spin-vector Monte Carlo \cite{Shin2014}. In particular, SQA occupies a unique position in that it simulates, as its name suggests, thermal equilibrium properties of the transverse-field Ising model  and thus may be regarded as an efficient tool to classically reproduce the outputs of QA provided that simulation is carried out at sufficiently low temperature. This is indeed true as long as static (time independent) properties of the transverse-field Ising model are concerned.  However, the dynamics of SQA is governed by classical stochastic (Monte Carlo) processes, which is fundamentally different from the Schr\"odinger dynamics of QA, and therefore SQA is unable to simulate dynamical aspects of QA \cite{Bando2021}, the latter being one of the important topics of recent research activities in the field \cite{Crosson2020}.

Another constraint of SQA is that the sign problem arises if we include additional terms in the Hamiltonian of the transverse-field Ising model, e.g., terms quadratic in the $x$-component of Pauli matrix with a positive coefficient (to be denoted as $+XX$ and to be called a non-stoquastic catalyst), which sometimes lead to enhanced performance of QA \cite{Seki2012,Seoane2012,Seki2015,Nishimori2017,Albash2019,Takada2019}.  The matrix representation of such a Hamiltonian in the computational basis has non-zero off-diagonal elements with both positive and negative signs and may be regarded to be non-stoquastic
\footnote{
Strictly speaking, non-stoquasticity means that no local basis transformation can eliminate the existence of a positive sign in the off-diagonal elements \cite{Bravyi2009}. In the present paper, we adopt a more restricted criterion to use the computational basis \cite{Albash2019,Takada2019} because this is the conventional basis to represent the Hamiltonian in a practical implementation of SQA using the Suzuki-Trotter formula \cite{Suzuki}.
}.

For non-stoquastic Hamiltonians, SQA is generally considered to be unable to efficiently reproduce even equilibrium properties due to the sign problem coming from negative values of local Boltzmann factors in the Suzuki-Trotter representation of the partition function \cite{Suzuki}.  Nevertheless, we show in the present paper that, if we impose a restriction on possible values of the parameters in the Hamiltonian, local Boltzmann factors can be chosen non-negative under an appropriate choice of the Suzuki-Trotter representation and thus the sign problem may potentially be circumvented. Then we derive a condition for those parameters to satisfy in order for SQA to converge to thermal equilibrium in the long-time limit by generalizing a theory developed in our previous paper for the case without a non-stoquastic catalyst \cite{Kimura202209}.

This paper is structured as follows. In section \ref{sec2.1}, we deduce a condition on the Ising model with a non-stoquastic catalyst that ensures that local Boltzmann factors can be chosen non-negative. We also derive an explicit expression of the corresponding classical Ising model in higher dimensions. In section \ref{sec2.2}, we recapitulate the formulation of the master equation, which governs the SQA dynamics, in terms of the corresponding imaginary-time Schr\"odinger equation as discussed in \cite{Nishimori2014}.  The adiabatic condition for the imaginary-time Schr\"odinger equation is discussed in section \ref{sec3}.
In section \ref{sec4}, we deduce conditions for convergence of SQA to thermal equilibrium in the long-time limit. The paper is concluded in section \ref{sec5}.

\section{Ising model with the transverse interaction term}
\label{sec2}
\subsection{Formulation and sign problem}
\label{sec2.1}
We consider the transverse-field Ising model with a $+XX$ term and analyze its effect on the sign of local Boltzmann factors as well as on the convergence condition to thermal equilibrium under SQA. 

The model we study is defined by the Hamiltonian:
\begin{equation}
\label{Hamiltonian 1}
H(t)=H_{\rm Ising}+H_{\rm TF}(t)+H_{\rm TI}(t).
\end{equation}
Here $ H_{\rm Ising}$  represents the Hamiltonian of the classical Ising model 
\begin{equation}
H_{\rm Ising}=-\sum_{\langle ij\rangle} J_{ij}\sigma^z_i\sigma^z_j,
\label{eq:Ising_H}
\end{equation}
where the interaction $J_{ij}$ is left arbitrary in the present paper. The summation runs over pairs of interacting spins $\langle ij\rangle$. $ H_{\rm TF}(t)$ denotes the transverse field term with a time-dependent coefficient $\Gamma(t) (\ge 0)$:
\begin{equation}
H_{\rm TF}=-\Gamma(t)\sum_{j=1}^N \sigma^x_j,
\end{equation}
and $H_{\rm TI}(t)$ is for the transverse interaction term $(+XX)$
\begin{equation}
\label{AFF term}
H_{\rm TI}(t)=K(t)\sum_{\langle ij\rangle} \sigma_i^x\sigma_j^x,
\end{equation}
where the time-dependent coefficient is non-negative $K(t)\ge 0$.

In SQA, the Hamiltonian is rewritten in terms of a corresponding classical Ising model utilizing the Suzuki-Trotter formula \cite{Suzuki}. 
The system is replicated into $M$ layers in the Trotter direction. To find a way to circumvent the problem of negative local Boltzmann factors, it is important to treat the two terms $H_{\rm TF}$ and $H_{\rm TI}$ simultaneously for a Trotterized representation of the Hamiltonian.  The matrix element connecting the $k$th and the $(k+1)$st Trotter layers for these two contributions together is
\begin{equation}
\label{term_k_matrix}
\bra{\sigma_{k+1}}e^{\frac{\beta\Gamma(t)}{M}\sum_{j=1}^N \sigma^x_j+\sum_{\langle ij\rangle}\frac{-\beta K(t)}{M}\sigma_i^x \sigma_j^x}\ket{\sigma_k},
\end{equation}
where $\beta$ denotes the inverse temperature and $\sigma_k$ stands for a set of $N$ classical Ising variables in the $k$th Trotter layer.
For simplicity, we will use the notation 
\begin{equation}
a:=\frac{\beta\Gamma(t)}{bM},~
b_K:=\frac{\beta K(t)}{M},
\end{equation}
where $b$ represents the number of the nearest neighbors (not to be confused with $b_K$) to compensate for the overcounting of the transverse-field term in the formulation developed below. We will also write 
\begin{equation}
\sigma_1:=\sigma_i^x,~
\sigma_2:=\sigma_j^x
\end{equation}
in the following.

Out of the expression \eqref{term_k_matrix}, the local Boltzmann factor for the interacting pair $\langle ij\rangle$ has the following expansion:
\begin{flalign}
\label{eval_Boltzmann}
&e^{a\left(\sigma_1+\sigma_2\right)-b_K\sigma_1 \sigma_2} \\\nonumber
 =& \cosh^2a\, \cosh b_K-{\rm sinh}^2\, a\, \sinh  b_K+\cosh a\, \sinh  a\big(\cosh b_K-\sinh  b_K\big)(\sigma_1+\sigma_2)\\\nonumber
 +&\big(-\cosh^2 a\, \sinh  b_K+\cosh b_K\, {\rm sinh}^2\, a\big)\sigma_1\sigma_2.
\end{flalign} 
The coefficients of the first line on this right-hand side are clearly positive. The coefficient of the third term is non-negative when the following condition is satisfied:
\begin{equation}
\cosh b_K\, \sinh^2 a \ge \cosh^2 a\, \sinh b_K.
\end{equation}
This condition can be rewritten as follows:
\begin{equation}
\label{cond_coth_K}
\tanh^2\left(\frac{\beta \Gamma}{bM}\right) \ge
\tanh\left(\frac{\beta K}{M}\right) .
\end{equation}
Since the local Boltzmann factors for the original Ising Hamiltonian $H_{\rm Ising}$ are diagonal in the $\sigma_k$ basis and are trivially non-negative, we learn that the condition \eqref{cond_coth_K} is equivalent to the requirement that all local Boltzmann factors are non-negative. This implies that the transition matrix elements for the SQA dynamics are non-negative under this condition, and thus potentially circumventing the sign problem
\footnote{
Practical implementation of SQA involves additional considerations beyond evaluation of local Boltzmann factors including the problem of ergodicity and should therefore be carefully designed \cite{Mazzola2017}. 
}
. 
This is a consequence of the fact that we simultaneously treated $H_{\rm TF}$ and $H_{\rm TI}$ in equation \eqref{term_k_matrix}. Notice that the matrix elements (local Boltzmann factors) corresponding to $H_{\rm TI}$ alone can be negative when $K>0$.

The region where the condition \eqref{cond_coth_K} is satisfied is displayed as the shaded area in figure \ref{figregion}. In QA, we initially set $K=0$ and $\Gamma$ to a very large value, and then proceed to decrease $\Gamma$ toward zero. We simultaneously increase $K$ to a positive value and then decrease it toward zero because the final Hamiltonian should be the target classical Ising model $H_{\rm Ising}$. This protocol can be realized within the constraint of equation \eqref{cond_coth_K} as exemplified in terms of a dashed line in figure \ref{figregion}. 
\begin{figure}[ht]
\begin{center}
\includegraphics[width=.6\textwidth]{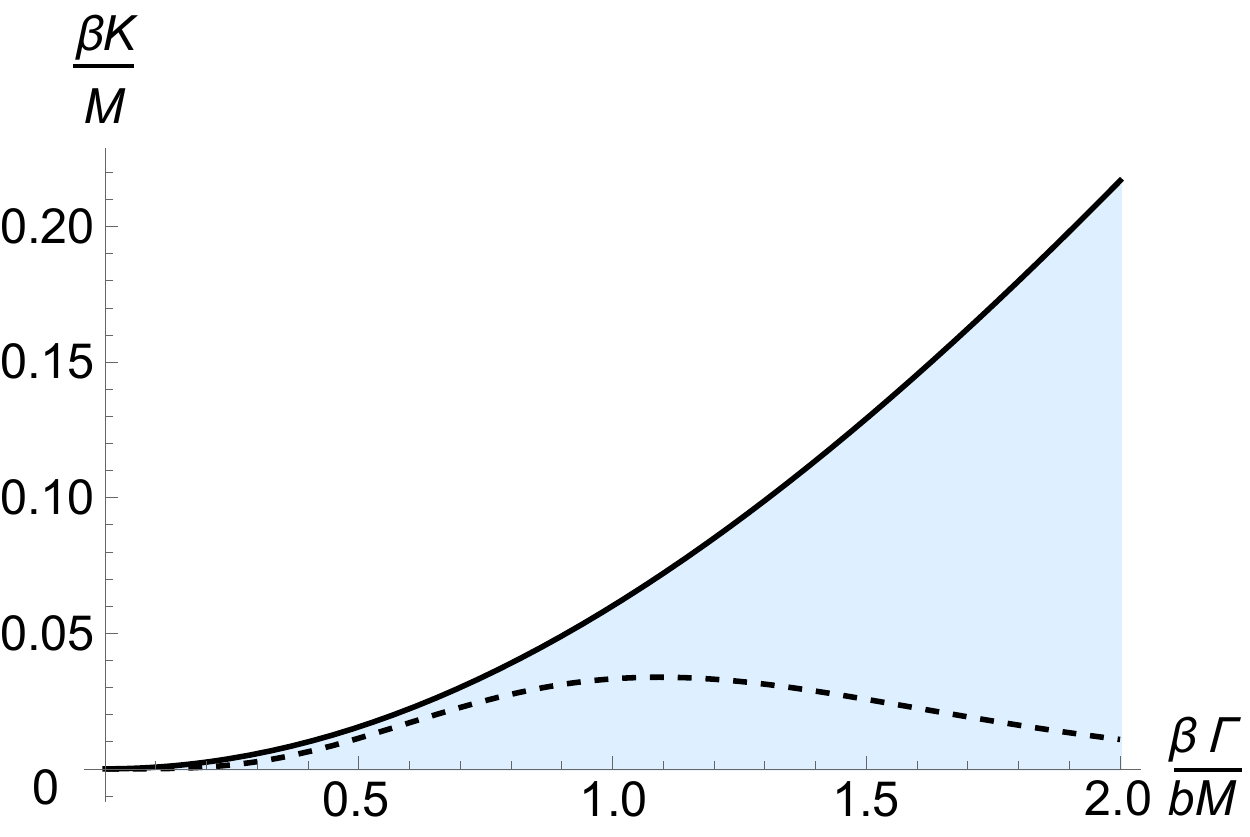}
\caption{\label{figregion}The shaded region satisfies the condition \eqref{cond_coth_K}. The solid curve describes $\tanh\left(\frac{\beta K}{M}\right)=\tanh^2\left(\frac{\beta \Gamma}{bM}\right)$. The shaded region below the solid curve is where the condition \eqref{cond_coth_K} is satisfied. A typical protocol to control the parameters in the shaded region is displayed as a dashed line. We set $\beta/M=1$ and $b=4$ in this figure.}
\end{center}
\end{figure}
SQA therefore would work consistently under this condition.

As detailed in \ref{appendix_1}, the total effective classical Ising model after Trotterization is given by the following expression:
\begin{flalign}
\label{Hamiltonian AFF}
&\beta H_0(\sigma)=\nonumber\\
&-\sum_{k=1}^M\sum_{\langle ij\rangle}\Bigg(\frac{\beta J_{ij}}{M} \sigma_i^{(k)}\sigma_j^{(k)}+ \alpha_1\left(\sigma^{(k+1)}_i\sigma_i^{(k)}+\sigma^{(k+1)}_j\sigma_j^{(k)}\right)+ \alpha_2 \sigma^{(k+1)}_i\sigma_i^{(k)}\sigma^{(k+1)}_j\sigma_j^{(k)}+\alpha_3 \Bigg),
\end{flalign}
where $\sigma$ denotes a configuration of $N$ Ising spins, i.e. $\sigma=\{\sigma_1, \sigma_2, \ldots, \sigma_N\}$, and $\alpha_1$, $\alpha_2$, $\alpha_3$ are  given as follows:
\begin{align}
\alpha_1 & =\frac{1}{4}\log \,\frac{e^{4c_1+2c_2}+e^{-4c_1-2c_2}+2}{e^{4c_1}+e^{2c_2}+e^{-4c_1}+e^{-2c_2}} \label{explicit_alpha1}\\
\alpha_2 & =\frac{1}{4}\log \,\frac{e^{4c_1}+e^{2c_2}+e^{-4c_1}+e^{-2c_2}}{e^{2c_2}+e^{-2c_2}+2} \label{explicit_alpha2}\\
\alpha_3 & =\frac{1}{4}\log \,(e^{4c_1}+e^{2c_2}+e^{-4c_1}+e^{-2c_2})(e^{2c_2}+e^{-2c_2}+2). \label{explicit_alpha3}
\end{align}
We have introduced the following notation:
\begin{equation}
\label{eq:c1c2}
 c_1:=\frac{1}{2}\log \coth a,\quad
 c_2:=\frac{1}{2}\log \coth\left(-b_K\right). 
\end{equation} 
If we write
\begin{equation}
d_2:= \frac{1}{2}\log \coth\left(b_K\right),
\end{equation}
where $d_2$ is real, 
\begin{equation}
\label{exp_c2_d2}
e^{2c_2}=e^{2d_2+i\pi}=-e^{2d_2}.
\end{equation}
Then the condition \eqref{cond_coth_K} can also be written as
\begin{equation}
\label{cond_d2}
e^{2d_2-4c_1}\ge 1.
\end{equation}
Using equation \eqref{exp_c2_d2}, the expression \eqref{explicit_alpha2} for $\alpha_2$ can be rewritten as
\begin{equation}
\label{explicit_alpha2_d2}
\alpha_2=\frac{1}{4}\log \,\frac{e^{2d_2}+e^{-2d_2}-e^{4c_1}-e^{-4c_1}}{e^{2d_2}+e^{-2d_2}-2}
\end{equation}
and the expression \eqref{explicit_alpha1} for $\alpha_1$ becomes
\begin{equation}
\label{explicit_alpha1_d2}
\alpha_1=\frac{1}{4}\log \,\frac{e^{4c_1+2d_2}+e^{-4c_1-2d_2}-2}{e^{2d_2}+e^{-2d_2}-e^{4c_1}-e^{-4c_1}}.
\end{equation}
The expression \eqref{explicit_alpha3} for $\alpha_3$ is likewise expressed as
\begin{equation}
\label{explicit_alpha3_d2}
\alpha_3 =\frac{1}{4}\log \,(e^{2d_2}+e^{-2d_2}-e^{4c_1}-e^{-4c_1})(e^{2d_2}+e^{-2d_2}-2).
\end{equation}
We learn from equations \eqref{explicit_alpha2_d2}, \eqref{explicit_alpha1_d2}, and \eqref{explicit_alpha3_d2} that equation \eqref{cond_d2} (equivalent to the condition \eqref{cond_coth_K}) implies that the coefficients $\alpha_1$, $\alpha_2$, and $\alpha_3$ are real.
Otherwise, those coefficients can become imaginary.

\subsection{Imaginary-time Schr\"odinger equation}
\label{sec2.2}
Let us next consider the following master equation for the Markovian dynamics of SQA based on the classical Ising model $\beta H_0(\sigma)$:
\begin{equation}
\label{master eqn}
\frac{d}{dt}P_\sigma(t)=\sum_{\tilde{\sigma}} W_{\sigma\tilde{\sigma}} P_{\tilde{\sigma}}(t).
\end{equation}
 The symbol $P_\sigma(t)$ represents the probability of the system in configuration $\sigma$ at $t$. We consider the transition matrix of size $2^N$, $\hat{W}$, whose $\sigma\tilde{\sigma}$-entry is $W_{\sigma\tilde{\sigma}}$,  $\hat{W}_{\sigma\tilde{\sigma}}= W_{\sigma\tilde{\sigma}}$. An off-diagonal entry of the transition matrix $\hat{W}$ is given by \cite{Nishimori2014} 
\begin{equation}
W_{\sigma\tilde{\sigma}} = w_{\sigma\tilde{\sigma}}e^{-\frac{\beta}{2}\left(H_0(\sigma)-H_0(\tilde{\sigma}) \right)} \hspace{0.375in} (\sigma\ne \tilde{\sigma}).
\end{equation}

We adopt the heat bath method with single-spin and two-spin flips reflecting equations \eqref{term_k_matrix}-\eqref{eval_Boltzmann}, and thus $w_{\sigma\tilde{\sigma}}$ is expressed as 
\begin{align}
w_{\sigma\tilde{\sigma}}=\frac{ 1}{e^{\frac{\beta}{2}\left(H_0(\tilde{\sigma})-H_0(\sigma) \right)}+e^{-\frac{\beta}{2}\left(H_0(\tilde{\sigma})- H_0(\sigma) \right)}},
\end{align}
where configurations $\sigma$ and $\tilde{\sigma}$ differ by either a single spin or two spins. Using the probability conservation condition $\sum_{\tilde{\sigma}}W_{\tilde{\sigma}\sigma}=0$, a diagonal entry of the transition matrix is given by 
\begin{equation}
W_{\sigma\sigma}=-\sum_{\tilde{\sigma}(\ne\sigma)}e^{-\frac{\beta}{2}\left(H_0(\tilde{\sigma}) -H_0(\sigma) \right)} w_{\tilde{\sigma}\sigma}.
\end{equation}

The master equation for SQA can be rewritten in terms of an imaginary-time Schr\"odinger equation \cite{Henley2004,Castelnovo2005,Nishimori2014}. The following equation yields the quantum Hamiltonian for  the imaginary-time Schr\"odinger equation:
\begin{equation}
\label{quantum Hamiltonian W}
\hat{H}:= -e^{\frac{\beta}{2}\hat{H}_0}\hat{W} e^{-\frac{\beta}{2}\hat{H}_0},
\end{equation}
where $\hat{H}_0$ denotes a diagonal matrix whose $\sigma\sigma$-entry is $H_0(\sigma)$, $(\hat{H}_0)_{\sigma\sigma}=H_0(\sigma)$.

One can rewrite the quantum Hamiltonian \eqref{quantum Hamiltonian W} as follows \cite{Nishimori2014}:
\begin{eqnarray}
\label{explicit qHamiltonian}
& \hat{H} = \sum_{\sigma}\sum_{\sigma'}w_{\sigma\sigma'} \left( e^{-\frac{\beta}{2}\left(H_0(\sigma') -H_0(\sigma) \right)}\ket{\sigma}\bra{\sigma}-\ket{\sigma'}\bra{\sigma} \right)\\\nonumber
& -\sum_{\sigma}\sum_{\sigma''}w_{\sigma\sigma''} \ket{\sigma''}\bra{\sigma}.
\end{eqnarray}
Here, configurations $\sigma$ and $\sigma'$ differ only at a single site, and configurations $\sigma$ and $\sigma''$ differ at two sites, corresponding to equation (\ref{eval_Boltzmann}).

Because configurations $\sigma$ and $\sigma'$ differ only at a single site, say they differ at the $i$th site on the $k$th Trotter slice, using the expression for Hamiltonian $\beta H_0(\sigma)$ \eqref{Hamiltonian AFF}, we find that the following expression yields the difference $\beta H_0(\sigma') - \beta H_0(\sigma)$:
\begin{eqnarray}
\label{difference Hamiltonians}
& \beta H_0(\sigma')- \beta H_0(\sigma)= \sum_{l \, ({\rm n.n.}\, i)} \Bigg(\frac{2\beta}{M}J_{il}\sigma_i^{(k)}\sigma_{l}^{(k)} \\\nonumber
& +2\alpha_2 \left(\sigma^{(k+1)}_i\sigma_i^{(k)}\sigma^{(k+1)}_l\sigma_l^{(k)}+\sigma^{(k)}_i\sigma_i^{(k-1)}\sigma^{(k)}_l\sigma_l^{(k-1)}\right)  \\\nonumber
& + 2\alpha_1\left(\sigma^{(k+1)}_i\sigma_i^{(k)}+\sigma^{(k)}_i\sigma_i^{(k-1)}\right)\Bigg).
\end{eqnarray}
We used the symbol ``$l\, ({\rm n.n.}\, i)$" to indicate summation over sites $l$ interacting with site $i$.

In a similar manner, since configurations $\sigma$ and $\sigma''$ differ at two sites, say they differ at the $i$th site and the $j$th site on the $k$th Trotter slice, the difference $\beta H_0(\sigma'') - \beta H_0(\sigma)$ is given by the following expression:
\begin{flalign}
\label{difference Hamiltonians_two}
& \beta H_0(\sigma'')- \beta H_0(\sigma)= \sum_{l \, ({\rm n.n.}\, i)} 2\alpha_1\left(\sigma^{(k+1)}_i\sigma_i^{(k)}+\sigma^{(k)}_i\sigma_i^{(k-1)}\right)\\\nonumber
& +\sum_{m \, ({\rm n.n.}\, j)} 2\alpha_1\left(\sigma^{(k+1)}_j\sigma_j^{(k)}+\sigma^{(k)}_j\sigma_j^{(k-1)}\right)\\\nonumber
& + \sum_{l \, ({\rm n.n.}\, i), l\ne j} \left(\frac{2\beta}{M}J_{il}\sigma_i^{(k)}\sigma_{l}^{(k)} +2\alpha_2 \left(\sigma^{(k+1)}_i\sigma_i^{(k)}\sigma^{(k+1)}_l\sigma_l^{(k)}+\sigma^{(k)}_i\sigma_i^{(k-1)}\sigma^{(k)}_l\sigma_l^{(k-1)}\right)\right)\\\nonumber
& + \sum_{m \, ({\rm n.n.}\, j), m\ne i} \left(\frac{2\beta}{M}J_{mj}\sigma_m^{(k)}\sigma_{j}^{(k)} +2\alpha_2 \left(\sigma^{(k+1)}_m\sigma_m^{(k)}\sigma^{(k+1)}_j\sigma_j^{(k)}+\sigma^{(k)}_m\sigma_m^{(k-1)}\sigma^{(k)}_j\sigma_j^{(k-1)}\right)\right).
\end{flalign}
We used $\sum_{l \, ({\rm n.n.}\, i), l\ne j}$ to represent the sum over sites that interact with the $i$th site, excluding the $j$th site. Because the spins are flipped at the $i$th and $j$th sites for the configuration $\sigma''$, the difference vanishes in the sum for the interacting pair $\langle i,j \rangle$. 

For the sake of brevity, the difference \eqref{difference Hamiltonians} is denoted as $-2\beta H_{i,k}$:
\begin{eqnarray}
    -\beta H_{i,k}:=\frac{1}{2}\big(\beta H_0(\sigma')-\beta H_0(\sigma)\big).
\end{eqnarray}
Similarly, we define
\begin{eqnarray}
\label{definition of difference_two}
& -\beta H_{i,j,k}:= \frac{1}{2}\big(\beta H_0(\sigma'')-\beta H_0(\sigma)\big).
\end{eqnarray}
We rewrite the diagonal coefficient in \eqref{explicit qHamiltonian} as:
\begin{equation}
\label{diag equality}
w_{\sigma\sigma'}e^{-\frac{\beta}{2}\left(H_0(\sigma')- H_0(\sigma) \right)} =\frac{ e^{\beta H_{i,k}}}{e^{\beta H_{i,k}}+ e^{-\beta H_{i,k}}}.
\end{equation}
The off-diagonal coefficients in \eqref{explicit qHamiltonian} are rewritten in an analogous fashion as
\begin{align}
\label{off-diag equality}
w_{\sigma\sigma'}  = \frac{ 1}{e^{\beta H_{i,k}}+e^{-\beta H_{i,k}}}, \quad
w_{\sigma\sigma''}  = \frac{ 1}{e^{\beta H_{i,j,k}}+e^{-\beta H_{i,j,k}}}
\end{align}

The classical master equation \eqref{master eqn} can also be written in the matrix-vector notations as
\begin{equation}
\label{master eqn rewritten}
\frac{d}{dt}\hat{P}(t)=\hat{W}(t) \hat{P}(t).
\end{equation}
Here, $\hat{P}(t)$ is used to represent a vector whose $\sigma$th entry is $P_\sigma(t)$, i.e. $(\hat{P}(t))_\sigma:=P_\sigma(t)$.
We introduce a wavefunction $\phi(t)$ as
\begin{equation}
\phi(t):= e^{\frac{\beta}{2}\hat{H}_0}\hat{P}(t),
\end{equation}
then as discussed in \cite{Nishimori2014} the master equation \eqref{master eqn rewritten} is rewritten as the following the imaginary-time Schr\"odinger equation:
\begin{equation}
\label{imaginary Sch}
-\frac{d\phi(t)}{dt}=\left(\hat{H}(t)-\frac{1}{2}\frac{d}{dt}(\beta \hat{H}_0) \right)\, \phi(t).
\end{equation}


\section{Adiabatic condition}
\label{sec3}
 Since the instantaneous ground state of the imaginary-time Schr\"odinger equation \eqref{imaginary Sch} corresponds to the thermal equilibrium state of the classical Ising Hamiltonian $H_0(\sigma)$ \cite{Nishimori2014}, which is the stationary state of the Master equation \eqref{master eqn}, the adiabatic condition that the system stays close to the instantaneous ground state during the evolution following equation \eqref{imaginary Sch} is equivalent to the condition that the system stays close to equilibrium in the Markovian dynamics of SQA. We notice that the adiabatic condition for the imaginary-time Schr\"odinger equation as discussed below is given in reference \cite{Morita2008}.

We introduce the notation $\hat{\mathcal{H}}(t)$ to represent the operator on the right-hand side of the imaginary-time Schr\"odinger equation \eqref{imaginary Sch}:
\begin{equation}
\mathcal{\hat{H}}(t) = \hat{H}(t) - \frac{1 }{2}\frac{ d}{dt}(\beta\hat{H}_0).
\end{equation}
According to reference \cite{Morita2008}, convergence of SQA to thermal equilibrium in the long-time limit is ensured  provided that the condition
\begin{equation}
\label{adiabatic condition model}
\frac{\left\lVert\frac{d\mathcal{\hat{H}}(t)}{dt}\right\rVert}{\Delta(t)^2}\ll 1.
\end{equation} 
is satisfied for sufficiently large $t$, where $\Delta (t)$ is the energy gap of the instantaneous ground state and the first excited state of $\mathcal{\hat{H}}(t)$ and $\lVert \cdots \rVert$ denotes the operator norm.

We first evaluate an upper bound on the norm of the operator in equation \eqref{adiabatic condition model} using equations \eqref{explicit qHamiltonian} and \eqref{Hamiltonian AFF},
\begin{flalign}
\label{ineq bound}
&\left\lVert\frac{d\mathcal{\hat{H}}(t)}{dt}\right\rVert\nonumber\\
&\le \ell MN \Bigg(\Bigg\lVert\frac{d}{dt}(w_{\sigma\sigma'} e^{-\frac{1}{2}\beta\left(H_0(\sigma')- H_0(\sigma) \right)})\Bigg\rVert+\Bigg\lVert \frac{dw_{\sigma\sigma'}}{dt} \Bigg\rVert +\frac{N-1}{2}\Bigg\lVert \frac{dw_{\sigma\sigma''}}{dt} \Bigg\rVert\Bigg) +\frac{1}{2}\Bigg\lVert \frac{d^2}{dt^2}(\beta \hat{H}_0) \Bigg\rVert,
\end{flalign}
where $\ell$ denotes a constant determined by the number of the sites interacting with a given site. 
Use of equation \eqref{diag equality} yields the following bound on the first term on the right-hand side of the inequality \eqref{ineq bound}: 
\begin{align}
\label{bound_part1}
& \Big\lVert \frac{d}{dt}(w_{\sigma\sigma'} e^{-\frac{\beta}{2}\left(H_0(\sigma')- H_0(\sigma) \right)}) \Big\rVert=\Bigg\lVert \frac{d}{dt}\left(\frac{e^{\beta H_{i,k}}}{ e^{\beta H_{i,k}}+ e^{-\beta H_{i,k}}}\right) \Bigg\rVert=\frac{1}{2}\Bigg\lVert \frac{d}{dt}(\beta H_{i,k})\cdot {\rm sech}^2(\beta H_{i,k})\Bigg\rVert\\ \nonumber
& \le\frac{1}{2}\Bigg\lVert \frac{d}{dt}(\beta H_{i,k})\Bigg\rVert.
\end{align}
Similarly, using equation \eqref{off-diag equality}, we find that
\begin{eqnarray}
\label{bound_part2}
& \Bigg\lVert \frac{d w_{\sigma\sigma'}}{dt} \Bigg\rVert = \Bigg\lVert \frac{d}{dt}\frac{1}{ e^{\beta H_{i,k}}+ e^{-\beta H_{i,k}}} \Bigg\rVert \\\nonumber
& = \frac{1}{2}\Bigg\lVert \frac{d}{dt}(\beta H_{i,k})\cdot {\rm tanh}(\beta H_{i,k}){\rm sech}(\beta H_{i,k}) \Bigg\rVert \le \frac{1}{4}\Bigg\lVert \frac{d}{dt}(\beta H_{i,k})\Bigg\rVert.
\end{eqnarray}
In analogous manners, we obtain the following bound:
\begin{align}
\label{bound_part3}
 \Bigg\lVert \frac{d w_{\sigma\sigma''}}{dt} \Bigg\rVert \le \frac{1}{4}\Bigg\lVert \frac{d}{dt}(\beta H_{i,j,k})\Bigg\rVert.
\end{align}
Furthermore, we obtain
\begin{eqnarray}
\label{bound_part4}
& \Bigg\lVert \frac{d^2}{dt^2}(\beta \hat{H}_0) \Bigg\rVert\le \sum_{k=1}^M\sum_{\langle ij\rangle}\Bigg( \frac{1}{2}\Bigg\lVert\left(\alpha_1\right)''\left(\sigma^{(k+1)}_i\sigma_i^{(k)}+\sigma^{(k+1)}_j\sigma_j^{(k)}\right)\Bigg\rVert \\\nonumber
&+ \Bigg\lVert\left(\alpha_2\right)'' \sigma^{(k+1)}_i\sigma_i^{(k)}\sigma^{(k+1)}_j\sigma_j^{(k)}\Bigg\rVert+\left\lvert\left(\alpha_3\right)''\right\rvert \Bigg)\\\nonumber
& \le M\sum_{\langle ij\rangle}\Big(\left\lvert\left(\alpha_1\right)''\right\rvert + \left\lvert\left(\alpha_2\right)''\right\rvert+\left\lvert\left(\alpha_3\right)''\right\rvert \Big).
\end{eqnarray} 

Next, we evaluate a lower bound of the energy gap $\Delta(t)$ between the ground state and the first excited state.  The coefficient $w_{\sigma\sigma''}$ in the quantum Hamiltonian \eqref{explicit qHamiltonian} has the following bound:
\begin{equation}
w_{\sigma\sigma''}=\frac{1}{e^{\beta H_{i,j,k}}+e^{-\beta H_{i,j,k}}}\ge \frac{1}{2 e^{\lvert\beta H_{i,j,k}\rvert}}.
\label{eq:w_bound}
\end{equation}
References \cite{Somma2007, Morita2007, Morita2008} derived a lower bound of the energy gap for a general Ising model. When the off-diagonal coefficients of the Hamiltonian \eqref{explicit qHamiltonian} are real and non-positive, Hopf's lemma as used in Theorem 2.4 of reference \cite{Morita2008} applies to the evaluation of the energy gap $\Delta(t)$. Since this condition of negative semi-definiteness of off-diagonal elements is equivalent to the condition that $\alpha_1$, $\alpha_2$, and $\alpha_3$ are real (and therefore $H_0$ is real), the method used in reference \cite{Morita2008} applies to the model under discussion given that the condition \eqref{cond_coth_K} is satisfied. 

For the Hamiltonian with two-spin-flip terms under discussion, one needs to determine  which of $N! (w_{\sigma\sigma'})^N$ or $(N/2)!(w_{\sigma\sigma''})^{N/2}$ is smaller according to reference \cite{Morita2008} because the smaller one gives a lower bound of the gap for $t$ larger than a certain threshold value (see discussions in section II.B.2 of \cite{Morita2008}).
To this end, we first note  that $\alpha_1$ approaches $c_1$, which grows indefinitely, and $\alpha_2$ approaches zero, according to equations (\ref{explicit_alpha1}) and (\ref{explicit_alpha2}). See also \ref{appendix_2}. 

Next, since 
\begin{equation}
\beta H_{i,j,k}= 2\beta H_{i,k}+\frac{2\beta}{M}J_{ij}\sigma_i^{(k)}\sigma_{j}^{(k)} +2\alpha_2 \left(\sigma^{(k+1)}_i\sigma_i^{(k)}\sigma^{(k+1)}_j\sigma_j^{(k)}+\sigma^{(k)}_i\sigma_i^{(k-1)}\sigma^{(k)}_j\sigma_j^{(k-1)}\right),
\end{equation} 
we have 
\begin{eqnarray}
w_{\sigma\sigma''} & = \frac{ 1}{e^{\beta H_{i,j,k}}+e^{-\beta H_{i,j,k}}} = \frac{ 1}{e^{2\beta H_{i,k}+2\delta}+e^{-2\beta H_{i,k}-2\delta}} \\\nonumber
& \le {\rm max}\{e^{2\delta}, e^{-2\delta}\}\, \frac{ 1}{e^{2\beta H_{i,k}}+e^{-2\beta H_{i,k}}}.
\end{eqnarray}
We used the notation
\begin{equation}
\delta:=\frac{\beta}{M}J_{ij}\sigma_i^{(k)}\sigma_{j}^{(k)} +\alpha_2 \left(\sigma^{(k+1)}_i\sigma_i^{(k)}\sigma^{(k+1)}_j\sigma_j^{(k)}+\sigma^{(k)}_i\sigma_i^{(k-1)}\sigma^{(k)}_j\sigma_j^{(k-1)}\right)
\end{equation}
for simplicity. We also use $\delta'$ to represent ${\rm max}\{e^{2\delta}, e^{-2\delta}\}$. Then, we obtain the following relation:
\begin{eqnarray}
\label{rel_w_1}
w_{\sigma\sigma''} & \le \frac{ \delta'}{e^{2\beta H_{i,k}}+e^{-2\beta H_{i,k}}}=\frac{ \delta'}{\left(e^{\beta H_{i,k}}+e^{-\beta H_{i,k}}\right)^2-2}\\\nonumber
& =\delta'\, \left(\left(w_{\sigma\sigma'}\right)^{-2}-2\right)^{-1}.
\end{eqnarray}
In the situation discussed in section \ref{sec4}, $w_{\sigma\sigma'}$ becomes sufficiently small as time $t$ becomes large. ($\delta$ approaches $\frac{\beta}{M}J_{ij}\sigma_i^{(k)}\sigma_{j}^{(k)}$ in the limit in the situation as in section \ref{sec4}.) Thus, for sufficiently large $t$, the following inequality holds 
\begin{equation}
\left(w_{\sigma\sigma'}\right)^{-2}-2\ge \left(\frac{1}{2}\right)^{2/N}\left(w_{\sigma\sigma'}\right)^{-2}.
\end{equation}
This inequality yields the following relation:
\begin{equation}
\label{rel_w_2}
2\left(w_{\sigma\sigma'}\right)^N\ge \left(\left(w_{\sigma\sigma'}\right)^{-2}-2\right)^{-N/2}.
\end{equation} 
Therefore, utilizing equations \eqref{rel_w_1} and \eqref{rel_w_2} we obtain the relation
\begin{eqnarray}
\label{rel_w_3}
\Big(\frac{N}{2}\Big)!(w_{\sigma\sigma''})^{N/2} & \le \Big(\frac{N}{2}\Big)!\, \delta'^{N/2}\left(\left(w_{\sigma\sigma'}\right)^{-2}-2\right)^{-N/2} \\\nonumber
& \le \Big(\frac{N}{2}\Big)!\, \delta'^{N/2}\,  2\left(w_{\sigma\sigma'}\right)^N \\\nonumber
& \le N!\left(w_{\sigma\sigma'}\right)^N.
\end{eqnarray}
We assumed that $N$ is sufficiently large to deduce the last line in equation \eqref{rel_w_3}. From this, one learns that $(N/2)!(w_{\sigma\sigma''})^{N/2}$ is smaller for large $t$ for the situation discussed in section \ref{sec4}. 

These arguments and the results obtained in \cite{Somma2007, Morita2007, Morita2008} applied to the model under discussion reveal that the following inequality holds for the energy gap $\Delta(t)$:
\begin{equation}
\label{delta_power}
\Delta(t) \ge A(N)( w_{\sigma\sigma''})^{N/2}.
\end{equation}
Owing to the lower bound of $w_{\sigma\sigma''}$ of equation \eqref{eq:w_bound}, we obtain the following lower bound for the energy gap $\Delta(t)$:
\begin{equation}
\label{delta_power_2}
\Delta(t) \ge\frac{A(N)}{2^{N/2}}e^{-N\lvert\beta H_{i,j,k}\rvert/2}.
\end{equation}
$A(N)$ is independent of $t$; however, it depends of the size of the system $N$ for $N\gg 1$ as
\begin{equation}
A(N)=a\sqrt{N}\, e^{-cN}.
\end{equation}
$a$ and $c$ are used to denote constants that do not depend on $N$ in the large-$N$ limit. As a result, one learns that the energy gap $\Delta(t)$ has the following lower bound:
\begin{equation}
\label{lower bound on Delta}
\Delta(t)\ge\frac{a\sqrt{N}}{2^{N/2}}e^{-N\lvert\beta H_{i,j,k}\rvert/2-cN}.
\end{equation}

Those obtained bounds in equations \eqref{ineq bound}-\eqref{bound_part4} and equation \eqref{lower bound on Delta} are similar to the corresponding bounds in the simpler case discussed in \cite{Kimura202209}, and therefore analyses of convergence condition developed therein apply without significant modifications.

\section{Convergence condition}
\label{sec4}
We closely follow reference \cite{Kimura202209} to derive a condition for convergence of SQA. 

It is useful to consider the case that $\Gamma(t)$ and $K(t)$ take the following forms for sufficiently large $t$:
\begin{align}
& \Gamma(t) = \frac{M}{\beta} {\rm tanh}^{-1}\left( \frac{1}{(c_3 t+c_4)^{g(t)}}\right) \label{function_form1} \\
& K(t) = \frac{M}{\beta} {\rm tanh}^{-1}\left( \frac{1}{(c_3 t+c_4)^{h(t)}}\right), \label{function_form2}
\end{align}
where $g(t)$ and $h(t)$ denote strictly positive and twice differentiable functions. $c_3$ and $c_4$ represent constants with $c_3>0$ strictly.  We now state our main result.

\begin{proposition}
\label{proposition1}
Simulated quantum annealing (SQA) with non-stoquastic transverse interactions converges  in the large-$t$ limit to thermal equilibrium at a given inverse temperature $\beta$, provided that $\Gamma(t)$ and $K(t)$ satisfy equation \eqref{cond_coth_K} and the function $g(t)$ and $h(t)$ in equations \eqref{function_form1} and \eqref{function_form2} satisfy the following conditions:
\begin{align}
    & \lim_{t\to\infty}\left(h(t)-2g(t)\right)\ne 0
    \label{propcondition_1},\\
    &0<g(t)\le \frac{1}{2N}
    \label{propcondition_2},\\
    &
    \lvert g'(t)\rvert \le \frac{c'}{(c_3 t+c_4) \log(c_3 t+c_4)},
    \label{propcondition_4}\\
    & \lvert g''(t)\rvert \le \frac{c''}{(c_3 t+c_4) \log(c_3 t+c_4)} 
    \label{propcondition_5}, \\
    &
    \lvert h'(t)\rvert \le d'
    \label{propcondition_h1},\\    
    & \lvert h''(t)\rvert \le \frac{d''}{(c_3 t+c_4) \log(c_3 t+c_4)} 
    \label{propcondition_h2}, \\
     & \lim_{t\to\infty}\frac{h(t)}{c_3 t+c_4}= 0
    \label{propcondition_3}, \\
    & \lim_{t\to\infty}h'(t)\log\, (c_3 t+c_4)= 0
    \label{propcondition_6}, \\
    &  \lim_{t\to\infty}h''(t)\log\, (c_3 t+c_4)= 0.    
    \label{propcondition_7}
\end{align}
 Here,  equations \eqref{propcondition_2}-\eqref{propcondition_h2} are to hold for sufficiently large $t$.  $c'$, and $c''$ denote positive constants, and $c_3$, $c'$, $c''$, $d'$, and $d''$ are to be selected sufficiently small. 
\end{proposition}

\begin{proof}

When $\Gamma(t)$ and $K(t)$ take the forms \eqref{function_form1} and \eqref{function_form2}, respectively, then $c_1$ and $c_2$ take the following forms:
\begin{align}
& c_1 = \frac{1}{2}\log (c_3 t+c_4)^{g(t)} \\
& d_2 = \frac{1}{2}\log (c_3 t+c_4)^{h(t)}. 
\end{align}
From this we deduce that 
\begin{equation}
\label{exp_h-2g}
e^{2d_2-4c_1}=(c_3 t+c_4)^{h(t)-2g(t)}.
\end{equation}
Compared to the model considered in reference \cite{Kimura202209}, the model under discussion contains additional terms with the coefficients $\alpha_2$ and $\alpha_3$. Our strategy is to confirm that additional contributions of $\alpha_2$ and $\alpha_3$ to the term $\left\lVert\frac{d\mathcal{\hat{H}}(t)}{dt}\right\rVert/{\Delta(t)^2}$ indeed vanish for sufficiently large $t$, and thus those contributions do not affect the condition for convergence to thermal equilibrium derived in reference \cite{Kimura202209}.
 
As shown in \ref{appendix_2},
\begin{equation}
\alpha_2\to 0
\end{equation}
as $t\to \infty$.
Next, we confirm that $\left(\alpha_2\right)'\to 0$ as $t\to\infty$ under the stated conditions.
\begin{equation}
\left(\alpha_2\right)'=\frac{1}{4}\left(\frac{\left(e^{2d_2}+e^{-2d_2}-e^{4c_1}-e^{-4c_1}\right)'}{e^{2d_2}+e^{-2d_2}-e^{4c_1}-e^{-4c_1}}-\frac{\left(e^{2d_2}+e^{-2d_2}-2\right)'}{e^{2d_2}+e^{-2d_2}-2}\right).
\end{equation}
It is straightforward to see that, for $t\gg 1$,
\begin{align}
\label{deriv_alpha_1}
& \frac{\left(e^{2d_2}+e^{-2d_2}-e^{4c_1}-e^{-4c_1}\right)'}{e^{2d_2}+e^{-2d_2}-e^{4c_1}-e^{-4c_1}} \sim \frac{\left(e^{2d_2}-e^{4c_1}\right)'}{e^{2d_2}-e^{4c_1}}\\\nonumber
& = \frac{h'\log \,(c_3 t+c_4)+c_3h(c_3 t+c_4)^{-1}-2g'(c_3 t+c_4)^{2g-h}\log \,(c_3 t+c_4)-2c_3g(c_3 t+c_4)^{2g-h-1}}{1-(c_3 t+c_4)^{2g-h}}.
\end{align}
It is simple to confirm that this tends to zero as $t\to\infty$ under the conditions \eqref{propcondition_1}, \eqref{propcondition_2}, \eqref{propcondition_3}, \eqref{propcondition_4}, and \eqref{propcondition_6}. We need the condition \eqref{propcondition_1} to assure that the denominator in the last line in \eqref{deriv_alpha_1} does not vanish in the limit $t\to\infty$. This means that equation \eqref{cond_coth_K} should hold as a strict inequality.

One can also show that the other term
\begin{equation}
\frac{\left(e^{2d_2}+e^{-2d_2}-2\right)'}{e^{2d_2}+e^{-2d_2}-2}
\end{equation}
vanish in the limit $t\to\infty$ in an analogous manner. 
These demonstrate that
\begin{equation}
\left(\alpha_2\right)'\to 0
\end{equation}
as $t\to \infty$.

It remains to confirm that $\left(\alpha_2\right)''$ approaches zero as $t\to\infty$. 
\begin{align}
& \left(\alpha_2\right)''= \frac{1}{4}\left(\frac{\left(e^{2d_2}+e^{-2d_2}-e^{4c_1}-e^{-4c_1}\right)''}{e^{2d_2}+e^{-2d_2}-e^{4c_1}-e^{-4c_1}}-\left(\frac{\left(e^{2d_2}+e^{-2d_2}-e^{4c_1}-e^{-4c_1}\right)'}{e^{2d_2}+e^{-2d_2}-e^{4c_1}-e^{-4c_1}}\right)^2\right)\\\nonumber
& -\frac{1}{4}\left(\frac{\left(e^{2d_2}+e^{-2d_2}-2\right)''}{e^{2d_2}+e^{-2d_2}-2}-\left(\frac{\left(e^{2d_2}+e^{-2d_2}-2\right)'}{e^{2d_2}+e^{-2d_2}-2}\right)^2 \right).
\end{align}
However, we just saw that the terms 
\begin{equation}
\frac{\left(e^{2d_2}+e^{-2d_2}-e^{4c_1}-e^{-4c_1}\right)'}{e^{2d_2}+e^{-2d_2}-e^{4c_1}-e^{-4c_1}}, \hspace{.2in} \frac{\left(e^{2d_2}+e^{-2d_2}-2\right)'}{e^{2d_2}+e^{-2d_2}-2}
\end{equation}
both approach zero for $t\to\infty$; therefore, it suffices to check that the two terms 
\begin{equation}
\frac{\left(e^{2d_2}+e^{-2d_2}-e^{4c_1}-e^{-4c_1}\right)''}{e^{2d_2}+e^{-2d_2}-e^{4c_1}-e^{-4c_1}}, \hspace{.2in} \frac{\left(e^{2d_2}+e^{-2d_2}-2\right)''}{e^{2d_2}+e^{-2d_2}-2}
\end{equation}
tend to zero as $t\to\infty$. 

We show that 
\begin{equation}
\frac{\left(e^{2d_2}+e^{-2d_2}-2\right)''}{e^{2d_2}+e^{-2d_2}-2}\to 0
\end{equation}
as $t\to\infty$. As per reasoning similar to the previous argument, the term $e^{-2d_2}$ can be ignored, and 
\begin{align}
& \frac{\left(e^{2d_2}+e^{-2d_2}-2\right)''}{e^{2d_2}+e^{-2d_2}-2}\sim \frac{\left(e^{2d_2}-2\right)''}{e^{2d_2}-2}\\\nonumber
& = \frac{h''\log\, (c_3 t+c_4)+2c_3 h'(c_3 t+c_4)^{-1}-c_3^2h(c_3 t+c_4)^{-2}+\left(h'\log\, (c_3 t+c_4)+c_3h(c_3 t+c_4)^{-1}\right)^2}{1-2(c_3 t+c_4)^{-h}}.
\end{align}
This tends to zero as $t\to\infty$ under the conditions \eqref{propcondition_3}, \eqref{propcondition_6}, and \eqref{propcondition_7}.

One can show in an analogous fashion that the term 
\begin{equation}
\frac{\left(e^{2d_2}+e^{-2d_2}-e^{4c_1}-e^{-4c_1}\right)''}{e^{2d_2}+e^{-2d_2}-e^{4c_1}-e^{-4c_1}}
\end{equation} 
also approaches zero for $t\to\infty$ under the stated conditions. These confirm that $\left(\alpha_2\right)''$ approaches zero as $t\to\infty$. 

These computations showed that 
\begin{equation}
\label{alpha2_to_zero}
\alpha_2, \,\left(\alpha_2\right)', \,\left(\alpha_2\right)''\to 0
\end{equation}
as $t\to\infty$ under the stated conditions.

Owing to equations \eqref{ineq bound}, \eqref{bound_part1}, \eqref{bound_part2}, \eqref{bound_part3}, and \eqref{bound_part4}, we have the bound as the following:
\begin{align}
\left\lVert\frac{d\mathcal{\hat{H}}(t)}{dt}\right\rVert\le & \ell MN\left(\frac{3}{4}\Bigg\lVert \frac{d}{dt}(\beta H_{i,k})\Bigg\rVert+\frac{N-1}{8}\Bigg\lVert \frac{d}{dt}(\beta H_{i,j,k})\Bigg\rVert \right) \\\nonumber
& + \frac{M}{2}\sum_{\langle ij\rangle}\Bigg( \left|\left(\alpha_1\right)''\right| + \left|\left(\alpha_2\right)''\right|+\left\lvert\left(\alpha_3\right)''\right\rvert \Bigg).
\end{align}
The bound contains terms $|\left(\alpha_2\right)'|$ and $|\left(\alpha_2\right)''|$. Owing to equation \eqref{lower bound on Delta}, an upper bound on $1/\Delta(t)^2$ includes terms
\begin{equation}
e^{N\big\lvert\alpha_2 \left(\sigma^{(k+1)}_i\sigma_i^{(k)}\sigma^{(k+1)}_l\sigma_l^{(k)}+\sigma^{(k)}_i\sigma_i^{(k-1)}\sigma^{(k)}_l\sigma_l^{(k-1)}\right)\big\rvert}
\end{equation}
and 
\begin{equation}
e^{N\big\lvert\alpha_2 \left(\sigma^{(k+1)}_m\sigma_m^{(k)}\sigma^{(k+1)}_j\sigma_j^{(k)}+\sigma^{(k)}_m\sigma_m^{(k-1)}\sigma^{(k)}_j\sigma_j^{(k-1)}\right)\big\rvert}.
\end{equation}
From equation \eqref{alpha2_to_zero} one learns that all these contributions from $\alpha_2$ (and the derivatives) vanish in the limit $t\to\infty$. 

This particularly means that the quartic term with the coefficient $\alpha_2$ in the Hamiltonian \eqref{Hamiltonian AFF} does not contribute to the upper bound on the term $\lim_{t\to\infty}\left\lVert\frac{d\mathcal{\hat{H}}(t)}{dt}\right\rVert/{\Delta(t)^2}$ in the adiabatic condition \eqref{adiabatic condition model} evaluated in section \ref{sec3}. 

The term $\lim_{t\to\infty}\frac{\left\lVert\frac{d\mathcal{\hat{H}}(t)}{dt}\right\rVert}{\Delta(t)^2}$ in the adiabatic condition also contains terms 
\begin{equation}
\frac{\left(\alpha_3\right)''}{\Delta(t)^2},
\end{equation}
and we need to confirm that these terms can be made arbitrarily small, to conclude the proof. 

We saw that $\alpha_2\to 0$ as $t\to\infty$; therefore, in the limit $t\to\infty$ the time-dependent factor that the lower bound of $\Delta(t)$ \eqref{lower bound on Delta} contains is $e^{-2N\alpha_1}$. Since $c_1, d_2\to\infty$ as $t$ approaches infinity, $\alpha_1$ and $\alpha_3$ tend to the following values as $t\to\infty$:
\begin{align}
\alpha_1\to \frac{1}{4}\log \,\frac{e^{4c_1}}{1-e^{4c_1-2d_2}}=c_1-\frac{1}{4}\log\,(1-e^{4c_1-2d_2}), \\\nonumber
\alpha_3\to \frac{1}{4}\log \,(e^{2d_2}-e^{4c_1})e^{2d_2}=d_2+\frac{1}{4}\log\,(1-e^{4c_1-2d_2}).
\end{align}
Note that the condition \eqref{propcondition_1} and the condition \eqref{cond_coth_K} (which is equivalent to condition \eqref{cond_d2}) imply that $1-e^{4c_1-2d_2}$ is strictly less than one. Owing to the conditions \eqref{propcondition_1} and \eqref{cond_coth_K}, utilizing equation \eqref{exp_h-2g} we find that $1-e^{4c_1-2d_2}$ converges to a strictly positive number as $t\to\infty$. Thus, $\log\,(1-e^{4c_1-2d_2})$ converges to a finite number as $t\to\infty$, which is negligible compared to $c_1$ and $d_2$. Therefore, it suffices to confirm that 
\begin{equation}
d''_2\, e^{4N c_1}
\end{equation}
can be made arbitrarily small.

Because 
\begin{equation}
d_2 = \frac{h(t)}{2}\log (c_3 t+c_4)
\end{equation}
and
\begin{equation}
\big\lvert e^{4N c_1}\big\rvert = \big\lvert c_3 t+c_4\big\rvert^{2Ng(t)}\le \big\lvert c_3 t+c_4\big\rvert,
\end{equation}
we have 
\begin{equation}
\label{bound_on_c2}
\big\lvert d''_2\, e^{4N c_1}\big\rvert\le \frac{1}{2} \left(\frac{-c^2_3 h(t)}{c_3 t+c_4}+h''(t)(c_3 t+c_4)\log (c_3 t+c_4)+2c_3 h'(t) \right).
\end{equation}
The bound on the right-hand side of \eqref{bound_on_c2} can be made arbitrarily small under the conditions \eqref{propcondition_2}, \eqref{propcondition_h1}, \eqref{propcondition_h2}, and \eqref{propcondition_3}. 

Thus the condition of convergence derived in our previous work \cite{Kimura202209} for the case without the $\alpha_2$ term, which is the inequalities \eqref{propcondition_2}, \eqref{propcondition_4}, and \eqref{propcondition_5}, also applies to the model under discussion. This concludes the proof. 
\end{proof}
\noindent {\it Remark 1}.\\
A simple example that satisfies the condition stated in Proposition \ref{proposition1} is 
\begin{align}
g(t)&=\frac{1}{2N}\\
h(t)&=\frac{2}{N}.
\end{align}
\noindent {\it Remark 2}.\\
Equation \eqref{propcondition_1} means that equation \eqref{cond_coth_K} should be satisfied as a strict inequality if we demand that SQA converges, not just that it produces non-negative local Boltzmann factors.

\vspace{2mm}
\noindent
{\it Remark 3}.\\
Equations \eqref{function_form1} and \eqref{function_form2} represent asymptotic expressions of those coefficients for sufficiently large $t$.  For small $t$, the initial condition should be respected, $\Gamma (0) \gg 1$ and $K(0)=0$. Also, at intermediate values of $t$, smaller values of the energy gap may exist than the one we evaluated in the asymptotic regime in the previous section and thus additional conditions may be required.
\vspace{2mm}

\noindent
{\it Remark 4}.\\
By choosing $\beta$ and $M$ sufficiently large, the system becomes sufficiently close to the ground state of the classical Ising model of equation \eqref{eq:Ising_H} in the long-time limit.

\section{Conclusion}
\label{sec5}
We have discussed the properties of SQA for the transverse-field Ising model with an additional non-stoquastic $+XX$ contribution to the Hamiltonian. This last term is known to lead to negative values of local Boltzmann factors if we compute its matrix elements for the Suzuki-Trotter decomposition in the computational basis. We have shown that this problem can be partially circumvented if we consider two contributions to matrix elements simultaneously, one from the transverse field and another from the $+XX$ term, not separately. 

In the parameter space where local Boltzmann factors are non-negative, we have derived conditions for SQA to converge to thermal equilibrium in the long-time limit. This is a generalization of our previous finding for the case without the $+XX$ contribution \cite{Kimura202209}.

It is known that the non-stoquastic $+XX$ term helps to enhance the performance of quantum annealing by reducing the order of phase transitions from first order to second order for mean-field-type problems \cite{Seki2012,Seoane2012,Seki2015,Nishimori2017,Takada2019,Albash2018}.  Unfortunately, in those problems, the shaded region of figure \ref{figregion}, where local Boltzmann factors are non-negative, vanishes in the thermodynamic limit  because of the long-range characteristics of $+XX$ interactions in the mean-field-type formulation in which $K$ is scaled by the system size as $K\to K/N$. It is an interesting topic of further work to see whether non-trivial practical implementation of SQA is possible for certain systems with a non-stoquastic $+XX$ contribution.

\ack
We thank Shunta Arai for useful comments.
This work is based on a project JPNP16007 commissioned by the New Energy and Industrial Technology Development Organization (NEDO).  

\appendix
\section{Derivation of the Trotterized Hamiltonian}
\label{appendix_1}
We derive equation \eqref{Hamiltonian AFF} utilizing the following equality:
\begin{eqnarray}
\label{ij_contribution}
& \bra{\sigma^{(k+1)}_i\sigma^{(k+1)}_j}e^{a\left(\sigma^x_i+\sigma^x_j\right)-b_K\sigma_i^x \sigma_j^x}\ket{\sigma^{(k)}_i\sigma^{(k)}_j}\\ \nonumber
& =\sum_{\sigma_1, \,\sigma_2=\pm 1} \bra{\sigma^{(k+1)}_i\sigma^{(k+1)}_j}e^{a\left(\sigma^x_i+\sigma^x_j\right)}\ket{\sigma_1\sigma_2} \bra{\sigma_1\sigma_2}e^{-b_K\sigma_i^x \sigma_j^x}\ket{\sigma^{(k)}_i\sigma^{(k)}_j}.
\end{eqnarray}
$\{\ket{\sigma_1\sigma_2}\}$ is used to denote an inserted complete basis. We will evaluate equation \eqref{ij_contribution}. The product of this over $\langle ij\rangle$ yields equation \eqref{term_k_matrix}. 

We first evaluate $\bra{\sigma_{k+1}}e^{-b_K\sigma_i^x \sigma_j^x}\ket{\sigma_k}$. We use $\sigma_i^{(k)}(=\pm 1)$ to represent a classical Ising spin on the $k$th Trotter slice at the $i$th site: 
\begin{eqnarray}
\label{exp_sigmasigma}
& \bra{\sigma_{k+1}}e^{-b_K\sigma_i^x \sigma_j^x}\ket{\sigma_k}\\\nonumber
& =A\left(\frac{\sigma_i^{(k)}\sigma_i^{(k+1)}+\sigma_j^{(k)}\sigma_j^{(k+1)}}{2}\right)^2\,e^{\frac{1}{4}\log {\rm coth}(-b_K)\left(\sigma_i^{(k)}\sigma_i^{(k+1)}+\sigma_j^{(k)}\sigma_j^{(k+1)}\right)},
\end{eqnarray}
where
\begin{align}
A=\sqrt{\frac{1}{2}{\rm sinh}\left(-2b_K\right)}.
\end{align}
The value of $\bra{\sigma^{(k+1)}_i}e^{a\sigma^x_i}\ket{\sigma_1}$ is computed as
\begin{equation}
\bra{\sigma^{(k+1)}_i}e^{a\sigma^x_i}\ket{\sigma_1}=\sqrt{\frac{1}{2}{\rm sinh}\left(2a\right)}\, e^{\frac{1}{2}\log\, {\rm coth}\, a\, \sigma^{(k+1)}_i\sigma_1};
\end{equation}
therefore,
\begin{equation}
\label{computation1}
\bra{\sigma^{(k+1)}_i\sigma^{(k+1)}_j}e^{a\left(\sigma^x_i+\sigma^x_j\right)}\ket{\sigma_1\sigma_2}=\frac{1}{2}{\rm sinh}\left(2a\right)\, e^{\frac{1}{2}\log\, {\rm coth}\, a \left(\sigma^{(k+1)}_i\sigma_1+\sigma^{(k+1)}_j\sigma_2\right)}.
\end{equation}
Thus, we deduce that 
\begin{flalign}
\label{comp_Hamiltonian}
& \bra{\sigma^{(k+1)}_i\sigma^{(k+1)}_j}e^{a\left(\sigma^x_i+\sigma^x_j\right)-b_K\sigma_i^x \sigma_j^x}\ket{\sigma^{(k)}_i\sigma^{(k)}_j}\\ \nonumber
& =\frac{1}{2}{\rm sinh}\left(2a\right) A \times\\\nonumber
& \sum_{\sigma_1, \sigma_2=\pm 1}\left(\frac{\sigma_1\sigma_i^{(k)}+\sigma_2\sigma_j^{(k)}}{2}\right)^2 e^{\frac{1}{2}\log\, {\rm coth}\, a \left(\sigma^{(k+1)}_i\sigma_1+\sigma^{(k+1)}_j\sigma_2\right)+\frac{1}{4}\log {\rm coth}\left(-b_K\right)\left(\sigma_i^{(k)}\sigma_1+\sigma_j^{(k)}\sigma_2\right)}\\ \nonumber
& = \frac{1}{2}{\rm sinh}\left(2a\right) A\Bigg(e^{\frac{1}{2}\log\, {\rm coth}\, a \left(\sigma^{(k+1)}_i\sigma_i^{(k)}+\sigma^{(k+1)}_j\sigma_j^{(k)}\right)+\frac{1}{4}\log {\rm coth}\left(-b_K\right)\left(\left(\sigma_i^{(k)}\right)^2+\left(\sigma_j^{(k)}\right)^2\right)}\\ \nonumber
& + e^{-\frac{1}{2}\log\, {\rm coth}\, a \left(\sigma^{(k+1)}_i\sigma_i^{(k)}+\sigma^{(k+1)}_j\sigma_j^{(k)}\right)-\frac{1}{4}\log {\rm coth}\left(-b_K\right)\left(\left(\sigma_i^{(k)}\right)^2+\left(\sigma_j^{(k)}\right)^2\right)} \Bigg) \\ \nonumber
& = \frac{1}{2}{\rm sinh}\left(2a\right) A\Bigg(e^{\frac{1}{2}\log\, {\rm coth}\, a \left(\sigma^{(k+1)}_i\sigma_i^{(k)}+\sigma^{(k+1)}_j\sigma_j^{(k)}\right)+\frac{1}{2}\log {\rm coth}\left(-b_K\right)}\\ \nonumber
& + e^{-\frac{1}{2}\log\, {\rm coth}\, a \left(\sigma^{(k+1)}_i\sigma_i^{(k)}+\sigma^{(k+1)}_j\sigma_j^{(k)}\right)-\frac{1}{2}\log {\rm coth}\left(-b_K\right)} \Bigg).
\end{flalign}
The last line in equation \eqref{comp_Hamiltonian} can be expressed as an exponential, i.e. we have:
\begin{eqnarray}
\label{rewritten_exp}
& e^{\frac{1}{2}\log\, {\rm coth}\, a \left(\sigma^{(k+1)}_i\sigma_i^{(k)}+\sigma^{(k+1)}_j\sigma_j^{(k)}\right)+\frac{1}{2}\log {\rm coth}\left(-b_K\right)} \\\nonumber
& + e^{-\frac{1}{2}\log\, {\rm coth}\, a \left(\sigma^{(k+1)}_i\sigma_i^{(k)}+\sigma^{(k+1)}_j\sigma_j^{(k)}\right)-\frac{1}{2}\log {\rm coth}\left(-b_K\right)}\\ \nonumber
& = e^{\alpha_1\left(\sigma^{(k+1)}_i\sigma_i^{(k)}+\sigma^{(k+1)}_j\sigma_j^{(k)}\right)+\alpha_2 \sigma^{(k+1)}_i\sigma_i^{(k)}\sigma^{(k+1)}_j\sigma_j^{(k)}+\alpha_3}
\end{eqnarray}
for appropriately chosen $\alpha_1, \alpha_2, \alpha_3$. This can be seen as follows: The left-hand side of equation \eqref{rewritten_exp} can be rewritten as 
\begin{eqnarray}
\label{expression1}
& (e^{c_2}+e^{-c_2}) {\rm cosh}^2\, c_1 \\\nonumber
& + (e^{c_2}-e^{-c_2}) \cosh c_1 \sinh  c_1 \left(\sigma^{(k+1)}_i\sigma_i^{(k)}+\sigma^{(k+1)}_j\sigma_j^{(k)}\right)\\\nonumber
& + (e^{c_2}+e^{-c_2}) {\rm sinh}^2\, c_1\, \sigma^{(k+1)}_i\sigma_i^{(k)}\sigma^{(k+1)}_j\sigma_j^{(k)}
\end{eqnarray}
with $c_1$ and $c_2$ as defined in equation \eqref{eq:c1c2}.
The right-hand side can be rewritten as 
\begin{eqnarray}
\label{expression2}
& e^{\alpha_1\left(\sigma^{(k+1)}_i\sigma_i^{(k)}+\sigma^{(k+1)}_j\sigma_j^{(k)}\right)+\alpha_2 \sigma^{(k+1)}_i\sigma_i^{(k)}\sigma^{(k+1)}_j\sigma_j^{(k)}+\alpha_3} \\\nonumber
& =e^{\alpha_3} ({\rm cosh}^2\, \alpha_1 \cosh\alpha_2+ {\rm sinh}^2\, \alpha_1 \sinh  \alpha_2) \\\nonumber
& + e^{\alpha_3} \cosh\alpha_1 \sinh  \alpha_1 (\cosh\alpha_2+ \sinh  \alpha_2)  \left(\sigma^{(k+1)}_i\sigma_i^{(k)}+\sigma^{(k+1)}_j\sigma_j^{(k)}\right) \\\nonumber
& + e^{\alpha_3} ({\rm sinh}^2\, \alpha_1 \cosh\alpha_2+ {\rm cosh}^2\, \alpha_1 \sinh  \alpha_2) \sigma^{(k+1)}_i\sigma_i^{(k)}\sigma^{(k+1)}_j\sigma_j^{(k)}.  
\end{eqnarray}
Equating the two expressions \eqref{expression1} and \eqref{expression2}, one finds that the equality in equation \eqref{rewritten_exp} holds when the following three conditions are satisfied:
\begin{flalign}
\label{three_conditions}
(e^{c_2}+e^{-c_2}) {\rm cosh}^2\, c_1 &=  e^{\alpha_3} ({\rm cosh}^2\, \alpha_1 \cosh\alpha_2+ {\rm sinh}^2\, \alpha_1 \sinh  \alpha_2) \\\nonumber
(e^{c_2}-e^{-c_2}) \cosh c_1 \sinh  c_1 &=  e^{\alpha_3} \cosh\alpha_1 \sinh  \alpha_1 (\cosh\alpha_2+ \sinh  \alpha_2) \\\nonumber
(e^{c_2}+e^{-c_2}) {\rm sinh}^2\, c_1 &=  e^{\alpha_3} ({\rm sinh}^2\, \alpha_1 \cosh\alpha_2+ {\rm cosh}^2\, \alpha_1 \sinh  \alpha_2).
\end{flalign}
The three conditions determine $\alpha_1, \alpha_2,$ and $\alpha_3$. When there are three parameters, $\alpha_1$, $\alpha_2$, and $\alpha_3$, the three conditions admit a solution. 

\section{Limits of coefficients}
\label{appendix_2}
We demonstrate that $\alpha_1$ grows indefinitely and $\alpha_2$ approaches zero as time $t$ tends toward infinity in the situation considered in section \ref{sec4}. 

Since both $K$ and $\Gamma$ approach zero as time $t$ tends toward infinity, given any $\varepsilon_K>0$ and $\varepsilon_{\Gamma}>0$ there is some time, $t_0$, such that for any time $t\ge t_0$, one has
\begin{equation}
\tanh\left(\frac{\beta K}{M}\right)<\varepsilon_K, \hspace{.2in} \tanh\left(\frac{\beta \Gamma}{bM}\right)<\varepsilon_{\Gamma}.
\end{equation}
For the sake of brevity, we utilize the following notations:
\begin{align}
t_K:= & \tanh\left(\frac{\beta K}{M}\right)\\
t_{\Gamma}:= & \tanh\left(\frac{\beta \Gamma}{bM}\right).
\end{align}
In these notations, 
\begin{align}
e^{4c_1} & =\left({\rm coth}\, \frac{\beta\Gamma}{bM}\right)^2=\frac{1}{t_{\Gamma}^2} \\
e^{2d_2} & ={\rm coth}\, \frac{\beta K}{M}=\frac{1}{t_K};
\end{align}
we then obtain
\begin{align}
\alpha_1 & = \frac{1}{4}\log \frac{\frac{1}{t_{\Gamma}^2t_K}+t_{\Gamma}^2t_K-2}{t_K+\frac{1}{t_K}-\frac{1}{t_{\Gamma}^2}-t_{\Gamma}^2}
\\\nonumber
& = \frac{1}{4}\log \frac{1- t_{\Gamma}^2t_K}{t_{\Gamma}^2-t_K}\ge \frac{1}{4}\log \frac{1- \varepsilon_{\Gamma}^2\varepsilon_K}{\varepsilon_{\Gamma}^2}.
\end{align}
Because $\frac{1}{4}\log \frac{1- \varepsilon_{\Gamma}^2\varepsilon_K}{\varepsilon_{\Gamma}^2}$ can be made arbitrarily large by selecting sufficiently small values for $\varepsilon_{\Gamma}$ and $\varepsilon_K$, one learns that $\alpha_1$ approaches infinity as $t\to\infty$. 

One also finds that
\begin{align}
\alpha_2 & = \frac{1}{4}\log \frac{t_K+\frac{1}{t_K}-t_{\Gamma}^2-\frac{1}{ t_{\Gamma}^2}}{t_K+\frac{1}{t_K}-2}=\frac{1}{4}\log \frac{t_{\Gamma}^2\left(1+t_K^2\right)-t_K\left(1+t_{\Gamma}^4\right)}{\left(1-t_K\right)^2t_{\Gamma}^2}\\\nonumber
& \le \frac{1}{4}\log \frac{1+t_K^2}{\left(1-t_K\right)^2}\\\nonumber
& = \frac{1}{4}\log\left(1+\frac{2t_K}{\left(1-t_K\right)^2}\right).
\end{align}
Because $\frac{t_K}{\left(1-t_K\right)^2}$ is a monotone increasing function of $t_K$ on the interval $[0,1)$, for $t_K<\varepsilon_K< 1$, one has the bound:
\begin{equation}
\alpha_2\le \frac{1}{4}\log\left(1+\frac{2\varepsilon_K}{\left(1-\varepsilon_K\right)^2}\right).
\end{equation}
The bound can be made arbitrarily small by choosing a sufficiently small value for $\varepsilon_K$. This shows that $\alpha_2$ approaches zero as $t\to\infty$.

\vspace{5mm}

\end{document}